\documentclass[journal]{IEEEtran}
\ifCLASSINFOpdf
\else
\fi

\usepackage{tikz}
\usepackage{amsmath}
\usepackage{amsfonts}
\usepackage{amssymb}
\usepackage{multirow}
\usepackage{amsthm}
\usepackage{cite}

\usepackage{ulem}
\usepackage{graphicx}

\newcommand{\specialcell}[2][c]{  \begin{tabular}[#1]{@{}c@{}}#2\end{tabular}}
\newtheorem{prp}{Proposition}
\begin{document}

\title{Texture Retrieval via the Scattering Transform}

\author{Alexander~Sagel,~Dominik~Meyer,~\IEEEmembership{Student~Member,~IEEE}~and~Hao~Shen,~\IEEEmembership{Member,~IEEE}
\thanks{The authors are with the Department of Electrical, Electronic 
and Computer Engineering, Technische Universit\"at M\"unchen, M\"unchen,
Germany. e-mail: \{a.sagel,dominik.meyer,hao.shen\}@tum.de.}
\thanks{This work has partially been supported by the Cluster of
Excellence \emph{CoTeSys}\;--\;Cognition for
Technical Systems, funded by the Deutsche
Forschungsgemeinschaft (DFG) and International Graduate School of 
Science and Engineering (IGSSE), Technische Universit\"at M\"unchen.}
}

\markboth{}%
{}
%

\maketitle

\begin{abstract}
This work studies the problem of content-based image retrieval,
specifically, texture retrieval. 
%
Our approach employs a recently developed method,
the so-called Scattering transform, for the process of feature 
extraction in texture retrieval.
It shares a distinctive property of providing a robust representation, 
which is stable with respect to spatial deformations.
Recent work has demonstrated its capability for texture classification,
and hence as a promising candidate for the problem of texture retrieval.
Moreover, we adopt a common approach of measuring the similarity of textures by comparing the subband histograms of a filterbank transform. To this end we derive a similarity measure based on the popular Bhattacharyya Kernel.
Despite the popularity of describing histograms using parametrized
probability density functions, such as the Generalized Gaussian Distribution, it is unfortunately not applicable 
for describing most of the Scattering transform subbands, due to the
complex modulus performed on each one of them. In this work, we 
propose to use the Weibull distribution to model the Scattering
subbands of descendant layers.
Our numerical experiments demonstrated the effectiveness of 
the proposed approach, in comparison with several state of 
the arts.
\end{abstract}

%
\IEEEpeerreviewmaketitle

\section{Introduction}
\subsection{Overview}
\label{sec:intro}

%
%
%
%
\IEEEPARstart{C}{ontent-based} image retrieval (CBIR) is a special 
case of image classification. It can be viewed as the process of assigning a query image to a set of image classes, where each class represents a database image.
%
%
Content-based hereby refers to the mode of formulating the search query. As opposed to metadata-based image search, which relies on pre-labeling the images beforehand, a CBIR system retrieves the $n$ best matches within the database with respect to the visual similarity to the query image.

In CBIR, the concepts of \textit{feature extraction} (FE) and \textit{similarity measure} (SM) play crucial roles. FE converts an image into a \textit{feature vector} of numerical values with the aim to produce a low-dimensional, yet sensible representation of the input image in the context of some particular application. An SM assigns a numerical value to a pair of two feature vectors. In this work, we assume that the SM is nonnegative and that a smaller SM value indicates a higher similarity and vice versa. A typical CBIR system is depicted in Fig. \ref{fig:cbir}. The images to be extracted are processed via an FE algorithm and the extracted feature vectors (signatures) are stored in a database. The query image is fed to the same FE algorithm and the output is compared with all the feature vectors in the database by means of the defined SM. As a result, the images with the lowest SM values are returned. Consequently, designing a CBIR system boils down to the construction of an FE/SM framework.

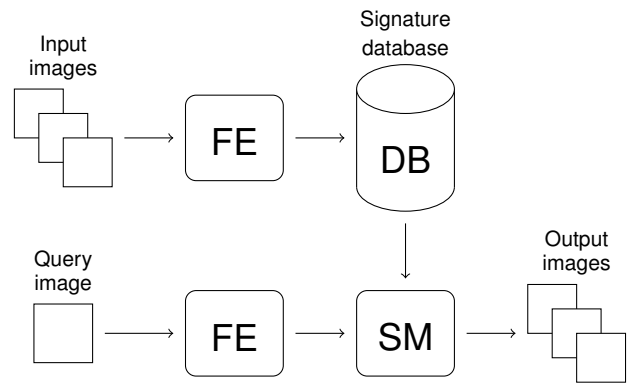
\begin{figure}
\begin{center}
\begin{tikzpicture}[scale=0.65]

\draw[fill=white] (8,4) ellipse (1 and 0.5);
\draw[fill=white, color=white] (7,4) rectangle (9,6);
\draw (7,3.99) -- (7,6);
\draw (9,3.99) -- (9,6);
\draw[fill=white] (8,6) ellipse (1 and 0.5);
\draw[fill=white] (0,5) rectangle (1,6);
\draw[fill=white] (0.5,4.5) rectangle (1.5,5.5);
\node[above,font=\footnotesize\sf] at (1,6.5) {Input}; 
\node[above,font=\footnotesize\sf] at (1,6) {images}; 
\draw[fill=white] (1,4) rectangle (2,5);
\draw[fill=white, rounded corners] (3.5,4.125) rectangle (5.5,5.875);
\node[above,font=\Large\sf] at (8,4.1) {DB}; 
\node[above,font=\Large\sf] at (4.5,4.45) {FE};
\node[above,font=\footnotesize\sf] at (8,7) {Signature}; 
\node[above,font=\footnotesize\sf] at (8,6.5) {database}; 

\draw[fill=white] (10.5,1) rectangle (11.5,2);
\draw[fill=white] (11,0.5) rectangle (12,1.5);
\draw[fill=white] (11.5,0) rectangle (12.5,1);
\node[above,font=\footnotesize\sf] at (11.5,2.5) {Output}; 
\node[above,font=\footnotesize\sf] at (11.5,2) {images}; 
\draw[fill=white, rounded corners] (7,0.125) rectangle (9,1.875);
\node[above,font=\Large\sf] at (8,0.45) {SM}; 
\draw[fill=white, rounded corners] (3.5,0.125) rectangle (5.5,1.875);
\node[above,font=\Large\sf] at (4.5,0.45) {FE}; 
\draw[fill=white] (0.4,0.4) rectangle (1.6,1.6);
\node[above,font=\footnotesize\sf] at (1,2.1) {Query}; 
\node[above,font=\footnotesize\sf] at (1,1.6) {image}; 

\draw[->] (2.25,5) -- (3.25,5);
\draw[->] (1.85,1) -- (3.25,1);
\draw[->] (5.75,5) -- (6.75,5);
\draw[->] (5.75,1) -- (6.75,1);
\draw[->] (9.25,1) -- (10.25,1);
\draw[->] (8,3.25) -- (8,2.125);

\end{tikzpicture}
\end{center}
\caption{A typical architecture of content-based image
retrival.}\label{fig:cbir}
\end{figure}

Nowadays, CBIR is employed to deal with big databases and vast amounts of data. Thus, in additional to high retrieval efficiency, execution speed is of great concern. This requires sufficient dimensionality reduction within the FE and a computationally efficient SM.

\subsection{Related Work and Our Contribution}
Extracting sensible information from images has been a field of research in digital signal processing for roughly half a century now. Naturally, the choice of the FE method depends on the specific application. For instance, when comparing images of landscapes, pictures taken in a forest will have different color distributions from those taken in a desert, so an FE based on the color properties of the image is a reasonable approach. On the other hand, images of objects are strongly characterized by shapes, so in this case the FE could rely on extracted contours. This work focuses on CBIR for textures. The earliest approaches for texture FE, such as the gray-level co-occurrence matrices (GLCM) \cite{Haralick1973} and the biologically motivated proposed Tamura features \cite{Tamura1978} have been presented in the 1970s and are often based on numerical measures for how nearby pixels relate to each other.

More recent approaches include incorporating spatial-frequency representations of the images \cite{Randen1999}, such as the Fast Wavelet Transform (FWT) \cite{Smith1994}, Gabor frames \cite{Manjunath1996} or Curvelet frames \cite{Sumana2008}. A conventional approach is designing the features based on the energy of the respective transform coefficients, but it can be beneficial to look at their statistical properties \cite{VandeWouwer1999,Vasconcelos2004}. Thus, FE approaches based on the histograms of filterbank transforms have become 
more prominent, cf. \cite{Do2002,Do2002a,Po2003,Tzagkarakis2006,Choy2010,Allili2012,Kwitt2008}. Techniques of this kind are referred to as 
\textit{Subband Histogram} methods in the work.

Linear filterbank transforms feed the input image to a bank of frequency-selective filters, yielding a set of band-pass signals as a representation. One problem of constructing FE based on such a decomposition is that the higher frequency subbands are prone to deformations in the spatial domain. Thereby, even slight deformations of an image will yield significantly different transform coefficients. However, simply neglecting the higher-frequency subbands is not desirable since they carry important information about such distinctive features as edges and corners. As a remedy, Mallat introduced the \textit{Scattering transform} \cite{Mallat2012}, a non-linear signal representation involving filterbanks and the modulus operation which transforms high-frequency components into low-pass representations. The Scattering transform appears to perform well in texture classification tasks \cite{Sifre2012,Sifre2013,Bruna2012}. The key contribution of this work is to define a statistical model in order to design subband histogram FE for texture images based on their Scattering transform coefficients. Furthermore, we introduce a similarity measure based on the extracted features, with the property of being a Kernel and thus can be interpreted as an inner product of Textures transformed to some Hilbert space.

The paper is organized as follows. 
Section~\ref{sec:math} provides necessary mathematical preliminaries 
of the problem of CBIR, subband histogram methods, and the
basics of the Scattering transform. 
In Section~\ref{sec:key_contrib} we develop a statistical framework 
of CBIR based on the Scattering transform.
In Section~\ref{sec:exp} several numerical experiments are presented
to investigate the performance of the proposed approach in comparison
with the state of the art methods.

\section{Mathematical Preliminaries}
\label{sec:math}

\subsection{Notations}
Unless stated otherwise, a \textit{signal} denotes an element of the Lebesgue space $L^2(\mathbb{R}^2)$. The terms \textit{almost everywhere} and \textit{almost all} refer to conditions which hold for all $\mathbb{R}^2\setminus \mathcal{N}$, where $\mathcal{N}$ can be any null set. For the sake of simplicity, we refer to $L^2(\mathbb{R}^2)$ as a Hilbert space of \textit{functions}, even though it would be more precise to call them \textit{equivalence classes of functions that are equal almost everywhere}. Bold-faced lowercase letters $\boldsymbol{x}$ or $\boldsymbol{x}_i$ describe vectors, while regular lowercase letters like $x$ or $x_i$ describe scalar values. Depending on the context, uppercase letters stand either for scalar values or for matrices. For a complex value $z$, $\bar{z}$ denotes its conjugate, and $\Re(z),\Im(z)$ the real and imaginary part, respectively. Angle brackets $\langle\cdot,\cdot\rangle$ denote the scalar product of $L^2(\mathbb{R}^2)$ and $\|\cdot\|$ the canonical norm induced by it. An asterisk denotes the convolution $f*g$ of two signals $f,g$.

\subsection{Subband Histogram Methods}
\label{sub:SubHist}
The coefficients $\{\boldsymbol{x}_1,...,\boldsymbol{x}_N\}$ produced by a transform of the query image can be viewed as a set of realizations of a random experiment. Let $p_d$ describe the probability density function (PDF) of the transform coefficients of some database image. 

The \textit{likelihood} of a set of random realizations $\{\boldsymbol{x}_1,...,\boldsymbol{x}_N\}$ for a PDF $p$ is defined as 
\begin{equation}
L(\boldsymbol{x}_1,...,\boldsymbol{x}_N|p)=\prod_{i=1}^N p(\boldsymbol{x}_i),
\label{eq:L}
\end{equation}
and is a measure for the probability that $\{\boldsymbol{x}_1,...,\boldsymbol{x}_N\}$ is subjected to $p$. That is to say, $L(\boldsymbol{x}_1,...,\boldsymbol{x}_N|p_d)$ can be viewed as a measure of similarity between the query image and the database images. Since the natural logarithm $\ln$ is a monotone function, this is also true for the \textit{log-likelihood} $\ln L$.

Let $\mathcal{A}$ denote the set of the PDFs of a number of database images. The best match for a query $\{\boldsymbol{x}_1,...,\boldsymbol{x}_N\}$ can thus be determined via
\begin{equation}
\begin{split}
p_d^{*}&=\operatorname*{\arg\max}_{p_d\in\mathcal{A}}\ln L(\boldsymbol{x}_1,...,\boldsymbol{x}_N|p_d)\\&=\operatorname*{\arg\max}_{p_d\in\mathcal{A}}\sum_{i=1}^N \ln p_d(\boldsymbol{x}_i).
\end{split}
\label{eq:ML}
\end{equation}
Expression \eqref{eq:ML} is called a \textit{Maximum Likelihood} (ML) solution.

Assume the query $\{\boldsymbol{x}_1,...,\boldsymbol{x}_N\}$ is subjected to a PDF $p_q$. Then, by the law of large numbers, the log-Likelihood $\ln L(\boldsymbol{x}_1,...,\boldsymbol{x}_N|p_d)$ can be approximated by the negative \textit{cross-entropy}
\begin{equation}
-H(p_q,p_d)=\int_{p_d(x)\neq0}p_q(x)\ln p_d(x)\mathrm{d}x.
\label{eq:xEntr}
\end{equation}
Assuming a generative model, the cross-entropy $H(p_q,p_d)$ defines a similarity measure between the respective query and database image.

One of the most well-studied transforms in image processing is the multiresolution decomposition produced by the FWT. In particular, it is known that the histograms of its band-pass subbands can be modeled by the \textit{Generalized Gaussian Distribution} (GGD), c.f.\cite{Mallat1989}, defined as 
\begin{equation}
p_\mathrm{GGD}(x|\alpha,\beta)=\frac{\beta}{2\alpha\Gamma(1/\beta)} \; e^{-(|x|/\alpha)^\beta},
\end{equation}
where $\Gamma$ denotes the Gamma-Function,
\begin{equation}
\Gamma(t) = \int_0^\infty  x^{t-1} e^{-x}\,{\rm d}x,
\end{equation}
and the parameter $\alpha\in\mathbb{R}_+$ determines the scale of the distribution while $\beta\in\mathbb{R}_+$ describes the shape. The cross-entropy between two GGDs ist given by
\begin{equation}
\begin{split}
&H(p_{\mathrm{GGD}}(x|\alpha_1,\beta_1),p_{\mathrm{GGD}}(x|\alpha_2,\beta_2))\\
=&\ln\frac{2\alpha_2}{\beta_2}+\ln\Gamma\left(\frac{1}{\beta_2}\right)
+\left(\frac{\alpha_1}{\alpha_2}\right)^{\beta_2}\frac{\Gamma\left((\beta_2+1)/\beta_1\right)}{\Gamma(1/\beta_1)}.
\end{split}
\label{eq:xEntr_GGD}
\end{equation}
Equation \eqref{eq:xEntr_GGD} provides a parametrized SM for texture images. The GGD parameters can be estimated from each image using FE based on ML which together with the SM in \eqref{eq:xEntr_GGD} defines a complete framework for CBIR. This approach was explored and discussed in \cite{Do2002}, though it was equivalently formulated in terms of the Kullback-Leibler divergence (KLD) rather than the cross-entropy. It can be viewed as a blueprint for other subband histogram methods for which the motivation is twofold. First, the FWT is not the only filterbank transform for images and is not necessarily the best basis for texture FE. Furthermore, even though the cross-entropy can be rigorously motivated by ML optimization, it lacks any geometrical interpretation, which poses the question, whether there are more suitable SMs for parametrized probability models.
\subsection{The Scattering transform}
\label{sub:WST}
\subsubsection{Definition}
Let $\theta\in L^2(\mathbb{R}^2)$ be a rotationally symmetric window function with low-pass characteristics. Let $\boldsymbol{\eta}\in\mathbb{R}^2\setminus\{\boldsymbol{0}\}$ and $J\in\mathbb{Z}$ be fixed and $\mathcal{R}\subset SO(2)$ a finite group of rotation matrices. With $\psi(\boldsymbol{x})=\theta(\boldsymbol{x})e^{i\boldsymbol{\eta}^\intercal\boldsymbol{x}}$, we define the  wavelet $\psi_{j, R}$ as
\begin{equation}
\psi_{j,R}(\boldsymbol{x})=4^{-j}\psi(2^{-j}R\boldsymbol{x}),\ j\in\{J,J-1,...\},R\in\mathcal{R}.
\end{equation}
Further, let us assume a low-pass and rotationally symmetric \textit{scaling function} \cite{Mallat2008} $\phi\in L^{2}(\mathbb{R}^2)$ and define
\begin{equation}
\phi_J(\boldsymbol{x})=4^{-J}\phi(2^{-J}\boldsymbol{x}),
\end{equation}
such that for the respective Fourier transforms $\hat{\psi},\hat{\phi}$,
\begin{equation}
|\hat{\phi}(2^J\boldsymbol{\omega})|^2+\sum\limits_{j=-\infty}^{J}\sum_{R\in\mathcal{R}}|\hat{\psi}(2^JR\boldsymbol{\omega})|^2=1
\label{eq:WST_tight}
\end{equation}
holds, for almost all $\boldsymbol{\omega}\in\mathbb{R}^2$. Then the set
\begin{equation}
\begin{split}
 \mathcal{D}=&\ \{\bar{\phi}_J(\boldsymbol{u}-\boldsymbol{x})\}_{\boldsymbol{u}\in\mathbb{R}^2}\\
&\cup\{\bar{\psi}_{j,R}(\boldsymbol{u}-\boldsymbol{x})\}_{\boldsymbol{u}\in\mathbb{R}^2,j\in\mathbb{Z}\wedge j\leq J,R\in\mathcal{R}}
\end{split}
\end{equation}
constitutes a Parseval-tight frame \cite{Mallat2008}. It will be assumed to span $L^2(\mathbb{R}^2)$ in the following. Note that for two signals $f,g$, the equation
\begin{equation}
(f*g)(\boldsymbol{u})=\langle f(\boldsymbol{x}),\bar{g}(\boldsymbol{u}-\boldsymbol{x})\rangle
\end{equation}
holds.

At the core of the \textit{Windowed Scattering transform} (WST) \cite{Mallat2012} is a dyadic wavelet decomposition $U_{\phi,\psi,J,\mathcal{R}}[f;j,R]$ of the input signal $f\in L^2 (\mathbb{R}^2)$ with the complex modulus performed on the band-pass components, defined as
\begin{equation}
U_{\phi,\psi,J,\mathcal{R}}[f;j,R]=\begin{cases}
|\psi_{j,R}*f|,\ &j\leq J,\\
\phi_J*f,\ &j>J.
\end{cases}
\label{eq:U}
\end{equation}
The modulus operation $|\cdot|$ traverses some of the energy of the band-pass signals towards lower frequencies. Therefore, $U_{\phi,\psi,J,\mathcal{R}}$ can be applied to the output signals $|\psi_{j,R}*f|$ again. This yields a tree structure like the one depicted in Fig. \ref{fig:WST_tree}. Note that \eqref{eq:U} yields an infinite (but countable) number of output signals. However, in practice we deal with bandlimited input signals $f$. Hence, there 
exists an integer $J_{l}$ with
$J_{l} <J$, such that Eq.~\eqref{eq:U} needs to be evaluated for $j\geq J_{l}$ only, which corresponds to a tree with finite breadth.

Basically, the idea of the WST is to apply $U_{\phi,\psi,J,\mathcal{R}}$ successively to the input signal and only keep the low-pass signals, i.e. to neglect signals represented by the black nodes in Fig. \ref{fig:WST_tree}.
\begin{figure*}
\begin{center}
\begin{tikzpicture}[scale=0.375]
\node[font=\small] at (24,19.5) {$f$};
\draw[->] (24,19) -- (24,18);
\draw[fill=white, rounded corners] (22.5,16.6) rectangle (25.5,18);
\node[font=\scriptsize] at (24,17.3) {$U_{\phi,\psi,J,\mathcal{R}}$};
\draw[->] (22.5,17.3) -- (21.5,17.3);
\draw[fill=white] (21,17.3) circle (0.5);
\draw[fill=black] (36,16) circle (0.05);
\draw[fill=black] (36.25,16) circle (0.05);
\draw[fill=black] (36.5,16) circle (0.05);
\draw[fill=black] (24,3) circle (0.05);
\draw[fill=black] (24,2.75) circle (0.05);
\draw[fill=black] (24,2.5) circle (0.05);
  \draw[fill=black] (8,14.5) circle (0.5);
  \draw[->] (23,16.6) -- (8.495,14.569);
  \draw[->] (8,14) -- (8,13);
  \draw[fill=white, rounded corners] (6.5,11.6) rectangle (9.5,13);
  \node[font=\scriptsize] at (8,12.3) {$U_{\phi,\psi,J,\mathcal{R}}$};
  \draw[->] (6.5,12.3) -- (5.5,12.3);
  \draw[fill=white] (5,12.3) circle (0.5);
    \draw[fill=black] (2.667,9.5) circle (0.5);
    \draw[->] (7,11.6) -- (3.132,9.683);
    \draw[->] (2.667,9) -- (2.667,8);
    \draw[fill=white, rounded corners] (1.167,6.6) rectangle (4.167,8);
    \node[font=\scriptsize] at (2.667,7.3) {$U_{\phi,\psi,J,\mathcal{R}}$};
    \draw[->] (1.167,7.3) -- (0.167,7.3);
    \draw[fill=white] (-0.333,7.3) circle (0.5);
    \draw[fill=black] (12,11) circle (0.05);
    \draw[fill=black] (12.25,11) circle (0.05);
    \draw[fill=black] (12.5,11) circle (0.05);
      \draw[fill=black] (0.889,4.5) circle (0.5);
      \draw[->] (1.667,6.6) -- (1.063,4.969);
      \draw[fill=black] (2.667,4.5) circle (0.5);
      \draw[->] (2.667,6.6) -- (2.667,5);
      \draw[fill=black] (4.444,4.5) circle (0.5);
      \draw[->] (3.667,6.6) -- (4.271,4.969);
      \draw[fill=black] (4.444,6) circle (0.05);
      \draw[fill=black] (4.694,6) circle (0.05);
      \draw[fill=black] (4.944,6) circle (0.05);
    \draw[fill=black] (8,9.5) circle (0.5);
    \draw[->] (8,11.6) -- (8,10);
    \draw[->] (8,9) -- (8,8);
    \draw[fill=white, rounded corners] (6.5,6.6) rectangle (9.5,8);
    \node[font=\scriptsize] at (8,7.3) {$U_{\phi,\psi,J,\mathcal{R}}$};
    \draw[->] (6.5,7.3) -- (5.5,7.3);
    \draw[fill=white] (5,7.3) circle (0.5);
      \draw[fill=black] (6.222,4.5) circle (0.5);
      \draw[->] (7,6.6) -- (6.396,4.969);
      \draw[fill=black] (8,4.5) circle (0.5);
      \draw[->] (8,6.6) -- (8,5);
      \draw[fill=black] (9.778,4.5) circle (0.5);
      \draw[->] (9,6.6) -- (9.604,4.969);
      \draw[fill=black] (9.778,6) circle (0.05);
      \draw[fill=black] (10.028,6) circle (0.05);
      \draw[fill=black] (10.278,6) circle (0.05);
    \draw[fill=black] (13.333,9.5) circle (0.5);
    \draw[->] (9,11.6) -- (12.869,9.683);
    \draw[->] (13.333,9) -- (13.333,8);
    \draw[fill=white, rounded corners] (11.833,6.6) rectangle (14.833,8);
    \node[font=\scriptsize] at (13.333,7.3) {$U_{\phi,\psi,J,\mathcal{R}}$};
    \draw[->] (11.833,7.3) -- (10.833,7.3);
    \draw[fill=white] (10.333,7.3) circle (0.5);    
      \draw[fill=black] (11.555,4.5) circle (0.5);
      \draw[->] (12.333,6.6) -- (11.729,4.969);
      \draw[fill=black] (13.333,4.5) circle (0.5);
      \draw[->] (13.333,6.6) -- (13.333,5);
      \draw[fill=black] (15.111,4.5) circle (0.5);
      \draw[->] (14.333,6.6) -- (14.937,4.969);
      \draw[fill=black] (15.111,6) circle (0.05);
      \draw[fill=black] (15.361,6) circle (0.05);
      \draw[fill=black] (15.611,6) circle (0.05);
  \draw[fill=black] (24,14.5) circle (0.5);
  \draw[->] (24,16.6) -- (24,15);
  \draw[->] (24,14) -- (24,13);
  \draw[fill=white, rounded corners] (22.5,11.6) rectangle (25.5,13);
  \node[font=\scriptsize] at (24,12.3) {$U_{\phi,\psi,J,\mathcal{R}}$};
  \draw[->] (22.5,12.3) -- (21.5,12.3);
  \draw[fill=white] (21,12.3) circle (0.5);
    \draw[fill=black] (18.667,9.5) circle (0.5);
    \draw[->] (23,11.6) -- (19.132,9.683);
    \draw[->] (18.667,9) -- (18.667,8);
    \draw[fill=white, rounded corners] (17.167,6.6) rectangle (20.167,8);
    \node[font=\scriptsize] at (18.667,7.3) {$U_{\phi,\psi,J,\mathcal{R}}$};
    \draw[->] (17.167,7.3) -- (16.167,7.3);
    \draw[fill=white] (15.667,7.3) circle (0.5);
    \draw[fill=black] (28,11) circle (0.05);
    \draw[fill=black] (28.25,11) circle (0.05);
    \draw[fill=black] (28.5,11) circle (0.05);
      \draw[fill=black] (16.889,4.5) circle (0.5);
      \draw[->] (17.667,6.6) -- (17.063,5);
      \draw[fill=black] (18.667,4.5) circle (0.5);
      \draw[->] (18.667,6.6) -- (18.667,5);
      \draw[fill=black] (20.444,4.5) circle (0.5);
      \draw[->] (19.667,6.6) -- (20.27,5);
      \draw[fill=black] (20.444,6) circle (0.05);
      \draw[fill=black] (20.694,6) circle (0.05);
      \draw[fill=black] (20.944,6) circle (0.05);
    \draw[fill=black] (24,9.5) circle (0.5);
    \draw[->] (24,11.6) -- (24,10);
    \draw[->] (24,9) -- (24,8);
    \draw[fill=white, rounded corners] (22.5,6.6) rectangle (25.5,8);
    \node[font=\scriptsize] at (24,7.3) {$U_{\phi,\psi,J,\mathcal{R}}$};
    \draw[->] (22.5,7.3) -- (21.5,7.3);
    \draw[fill=white] (21,7.3) circle (0.5);      
      \draw[fill=black] (22.222,4.5) circle (0.5);
      \draw[->] (23,6.6) -- (22.396,4.969);
      \draw[fill=black] (24,4.5) circle (0.5);
      \draw[->] (24,6.6) -- (24,5);
      \draw[fill=black] (25.778,4.5) circle (0.5);                
      \draw[->] (25,6.6) -- (25.64,4.969);
      \draw[fill=black] (25.778,6) circle (0.05);
      \draw[fill=black] (26.028,6) circle (0.05);
      \draw[fill=black] (26.278,6) circle (0.05);
    \draw[fill=black] (29.333,9.5) circle (0.5);
    \draw[->] (25,11.6) -- (28.869,9.683);   
    \draw[->] (29.333,9) -- (29.333,8);
    \draw[fill=white, rounded corners] (27.833,6.6) rectangle (30.833,8);
    \node[font=\scriptsize] at (29.333,7.3) {$U_{\phi,\psi,J,\mathcal{R}}$};
    \draw[->] (27.833,7.3) -- (26.833,7.3);
    \draw[fill=white] (26.333,7.3) circle (0.5);
      \draw[fill=black] (27.555,4.5) circle (0.5);
      \draw[->] (28.333,6.6) -- (27.729,4.969);
      \draw[fill=black] (29.333,4.5) circle (0.5);
      \draw[->] (29.333,6.6) -- (29.333,5);
      \draw[fill=black] (31.111,4.5) circle (0.5);
      \draw[->] (30.333,6.6) -- (30.937,4.969);
      \draw[fill=black] (31.111,6) circle (0.05);
      \draw[fill=black] (31.361,6) circle (0.05);
      \draw[fill=black] (31.611,6) circle (0.05);
  \draw[fill=black] (40,14.5) circle (0.5); 
  \draw[->] (25,16.6) -- (39.505,14.569);
  \draw[->] (40,14) -- (40,13);
  \draw[fill=white, rounded corners] (38.5,11.6) rectangle (41.5,13);
  \node[font=\scriptsize] at (40,12.3) {$U_{\phi,\psi,J,\mathcal{R}}$};
  \draw[->] (38.5,12.3) -- (37.5,12.3);
  \draw[fill=white] (37,12.3) circle (0.5);
    \draw[fill=black] (34.667,9.5) circle (0.5);
    \draw[->] (39,11.6) -- (35.132,9.683);
    \draw[->] (34.667,9) -- (34.667,8);
    \draw[fill=white, rounded corners] (33.167,6.6) rectangle (36.167,8);
    \node[font=\scriptsize] at (34.667,7.3) {$U_{\phi,\psi,J,\mathcal{R}}$};
    \draw[->] (33.167,7.3) -- (32.167,7.3);
    \draw[fill=white] (31.667,7.3) circle (0.5);
    \draw[fill=black] (44,11) circle (0.05);
    \draw[fill=black] (44.25,11) circle (0.05);
    \draw[fill=black] (44.5,11) circle (0.05);
      \draw[fill=black] (32.889,4.5) circle (0.5);
      \draw[->] (33.667,6.6) -- (33.063,4.969);
      \draw[fill=black] (34.667,4.5) circle (0.5);
      \draw[->] (34.667,6.6) -- (34.667,5);
      \draw[fill=black] (36.444,4.5) circle (0.5);
      \draw[->] (35.667,6.6) -- (36.27,4.969);
      \draw[fill=black] (36.444,6) circle (0.05);
      \draw[fill=black] (36.694,6) circle (0.05);
      \draw[fill=black] (36.944,6) circle (0.05);
    \draw[fill=black] (40,9.5) circle (0.5);
    \draw[->] (40,11.6) -- (40,10);
    \draw[->] (40,9) -- (40,8);
    \draw[fill=white, rounded corners] (38.5,6.6) rectangle (41.5,8);
    \node[font=\scriptsize] at (40,7.3) {$U_{\phi,\psi,J,\mathcal{R}}$};
    \draw[->] (38.5,7.3) -- (37.5,7.3);
    \draw[fill=white] (37,7.3) circle (0.5);        
      \draw[fill=black] (38.222,4.5) circle (0.5);
      \draw[->] (39,6.6) -- (38.396,4.969);
      \draw[fill=black] (40,4.5) circle (0.5);
      \draw[->] (40,6.6) -- (40,5);
      \draw[fill=black] (41.778,4.5) circle (0.5);
      \draw[->] (41,6.6) -- (41.604,4.969);        
      \draw[fill=black] (41.778,6) circle (0.05);
      \draw[fill=black] (42.028,6) circle (0.05);
      \draw[fill=black] (42.278,6) circle (0.05);
    \draw[fill=black] (45.333,9.5) circle (0.5);   
    \draw[->] (41,11.6) -- (44.833,9.683);
    \draw[->] (45.333,9) -- (45.333,8);
    \draw[fill=white, rounded corners] (43.833,6.6) rectangle (46.833,8);
    \node[font=\scriptsize] at (45.333,7.3) {$U_{\phi,\psi,J,\mathcal{R}}$};
    \draw[->] (43.833,7.3) -- (42.833,7.3);
    \draw[fill=white] (42.333,7.3) circle (0.5);        
      \draw[fill=black] (43.555,4.5) circle (0.5);
      \draw[->] (44.333,6.6) -- (43.729,4.969);
      \draw[fill=black] (45.333,4.5) circle (0.5);
      \draw[->] (45.333,6.6) -- (45.333,5);
      \draw[fill=black] (47.111,4.5) circle (0.5);
      \draw[->] (46.333,6.6) -- (46.937,4.969);
      \draw[fill=black] (47.111,6) circle (0.05);
      \draw[fill=black] (47.361,6) circle (0.05);
      \draw[fill=black] (47.611,6) circle (0.05);

\end{tikzpicture}
\end{center}
\caption{WST tree produced by successive application of $U_{\phi,\psi,J,\mathcal{R}}$ on the input signal $f$. Lowpass signals (The WST subbands) are depicted as white nodes. The complex moduli of band-pass signals are depicted as black nodes.}
\label{fig:WST_tree}
\end{figure*}
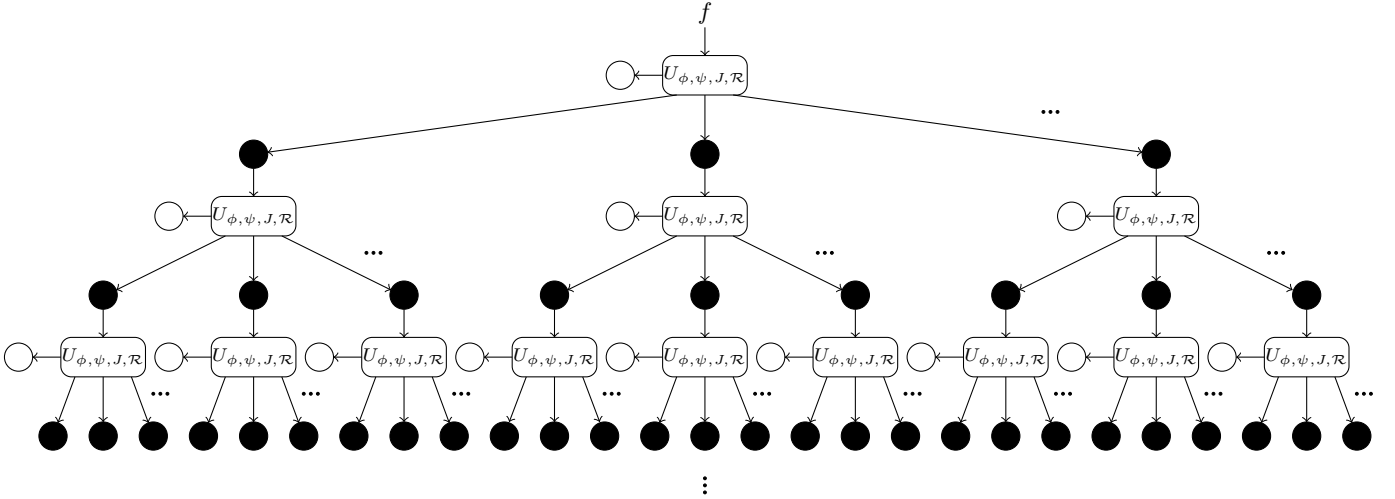
The WST along the \textit{path} $p=((j_1,R_1),...,(j_m,R_m))$ of scaling factors and rotations is defined as
\begin{equation}
S_{\phi,\psi,J,\mathcal{R}}[f;p]=\phi_J*|\psi_{j_m,R_m}*|\cdots*|\psi_{j_1,R_1}*f|\cdots||.
\end{equation}
Let us denote by $\mathcal{P}_M$ the (ordered) set of all possible 
paths $p$ of size $0$ to $M$.
We can then define the \textit{Finite-path WST} as
\begin{equation}
S_{\phi,\psi,J,\mathcal{R}}[f;\mathcal{P}_M]=\{S_{\phi,\psi,J,\mathcal{R}}[f;p]\}_{p\in\mathcal{P}_M}.
\label{eq:fp_wst}
\end{equation}
\subsubsection{Properties}
The output signals of the Finite-path WST are the elements of \eqref{eq:fp_wst}, also referred to as \textit{WST subbands} in the following, which are all low-pass signals due to the filtering with $\phi_J$. With increasing $M$, $S_{\phi,\psi,J,\mathcal{R}}[f;\mathcal{P}_M]$ captures more and more energy of the input signal $f$ and the WST representation becomes more expressive as the maximum path length grows. Let the Scattering norm be defined as
\begin{equation}
\left\|S_{\phi,\psi,J,\mathcal{R}}[f;\mathcal{P}_M]\right\|_\mathcal{S}^2=\sum_{p\in\mathcal{P}_M} \|S_{\phi,\psi,J,\mathcal{R}}[f;p] \|^2.
\end{equation}
As a consequence of the tight frame property \eqref{eq:WST_tight} of $\mathcal{D}$, the infinite-path WST is unitary, i.e.
\begin{equation}
\begin{split}
\|S_{\phi,\psi,J,\mathcal{R}}[f;\mathcal{P}_\infty]\|_\mathcal{S}&=\lim\limits_{M\rightarrow\infty}\|S_{\phi,\psi,J,\mathcal{R}}[f;\mathcal{P}_M]\|_\mathcal{S}\\&=\|f\|.
\end{split}
\end{equation}
In practice, maximum path depths of $M=3$ are expected to capture all essential information \cite{Bruna2013b,Anden2013}. Actually reconstructing $f$ from its WST involves a phase recovery problem. However, it is known \cite{Waldspurger2012} that digital realizations of wavelet representations similar to those employed in the WST are uniquely determined by their complex modulus.

The WST is non-expansive \cite{Mallat2012}, i.e. for two signals $f_1,f_2$, it can be shown \cite{Mallat2012} that for any positive integer $M$, 
\begin{equation}
\begin{split}
&\|S_{\phi,\psi,J,\mathcal{R}}[f_1;\mathcal{P}_M]-S_{\phi,\psi,J,\mathcal{R}}[f_2;\mathcal{P}_M]\|_\mathcal{S}^2\\
=&\sum\limits_{p\in\mathcal{P}_M}\|S_{\phi,\psi,J,\mathcal{R}}[f_1;p]-S_{\phi,\psi,J,\mathcal{R}}[f_2;p]\|^2\\
\leq&\|f_1-f_2\|^2
\end{split}
\end{equation}
holds, implying that the WST is robust with respect to additive noise.

Since WST is a representation consisting solely of low-pass signals, it is stable with respect to small spatial translations and deformations. 
Specifically, let $\tilde{f}$ be a deformed or translated version of $f$. Under certain mild assumptions about the underlying wavelet frame, it has been proven that the error $\|S_{\phi,\psi,J,\mathcal{R}}[f;\mathcal{P}_M]-S_{\phi,\psi,J,\mathcal{R}}[\tilde{f};\mathcal{P}_M]\|$ is bounded asymptotically, cf. \cite{Mallat2012}. 
%
Note, that we are given some freedom of choice in the wavelet frame which allows for considerable flexibility in terms of parameters such as frequency selectivity or directionality.

Due to its stability and flexibility, the WST is a convenient signal representation and can be used as the foundation for the FE from images.

\subsubsection{Normalized WST}
In order to reduce redundancy and to increase invariance to distortions, the \textit{Normalized WST} (NWST) \cite{Anden2013} was introduced. Let $\tilde{p}\in\mathcal{P}_M$ be the predecessor of $p\in\mathcal{P}_M$, i.e. $p=((j_1,R_1),...,(j_m,R_m))$ implies $\tilde{p}=((j_1,R_1),...,(j_{m-1},R_{m-1}))$.

Let $\varphi$ denote a very narrow-band lowpass blurring filter. For the layers $m\geq 1$, the NWST is defined as 
\begin{equation}
\bar{S}_{\varphi, \phi, \psi,J,\mathcal{R}}[f;p]=
\begin{cases}
\frac{S_{\phi, \psi,J,\mathcal{R}}[f;p]}{|f|*\varphi} & \mathrm{if\ } p\in\mathcal{P}_1, \\
\frac{S_{\phi, \psi,J,\mathcal{R}}[f;p]}{S_{\phi,\psi,J,\mathcal{R}}[f;\tilde{p}]} & \mathrm{otherwise.}
\end{cases}
\label{eq:nwst}
\end{equation}
In words, each subband of the WST is normalized by the respective parent subband, except for the subbands in the first layer which are normalized by the mean of the modulus of the input signal. In practice, a small constant is added to the denominator in order to avoid division by zero.

\section{Statistical Scattering CBIR}
\label{sec:key_contrib}
\subsection{Subband Modeling}
\label{sub:FE}
We propose to model the gray-value distributions of the different WST
subbands with parametrized PDFs and describe the images in terms of their 
respective parameters to obtain a complete FE mechanism on top of the WST.

The most distinctive features in textures are those of higher frequencies and are thus carried by the layers $m\geq1$. These layers contain signals of the form
\begin{equation}
\begin{split}
S_{\phi,\psi,J,\mathcal{R}}[f;p]=&\\
\phi_J*|\psi_{j_m,R_m}*|\cdots*|\psi_{j_1,R_1}*f|\cdots||&,\;p\neq p_0.
\end{split}
\label{eq:subband_m1}
\end{equation}
Clearly, the modulus eliminates all negative values which is why a symmetric model such as the GGD 
is not appropriate anymore. 
The histograms of signals from descendant layers suggest a PDF that occupies only positive values and can describe a skew to the left, which makes the \textit{Weibull Distribution} (WD) a nearby choice. In fact, the WD was already successfully employed in the modeling of complex wavelet coefficients \cite{Kwitt2008}. Fig. \ref{fig:hists} exemplarily shows histograms of WST subbands of different textures with a fitted WD curve. Despite variations in skewness and spread, the WD fittings describe the actual histograms fairly well. The PDF of the WD is defined for $x\geq 0$ as 
\begin{equation}
p_\mathrm{WD}(x|\lambda,k)=\lambda k \cdot (\lambda   x)^{k-1} \mathrm{e}^{-(\lambda \cdot x)^k}.
\label{eq:Weibull}
\end{equation}
Analogous to the GGD, $\lambda\in\mathbb{R}_+$ is the scale parameter, whereas $k\in\mathbb{R}_+$ determines the shape.
\begin{figure}
 \begin{center}
   \includegraphics[trim=23mm 8mm 21mm 5mm,clip,width=\columnwidth]{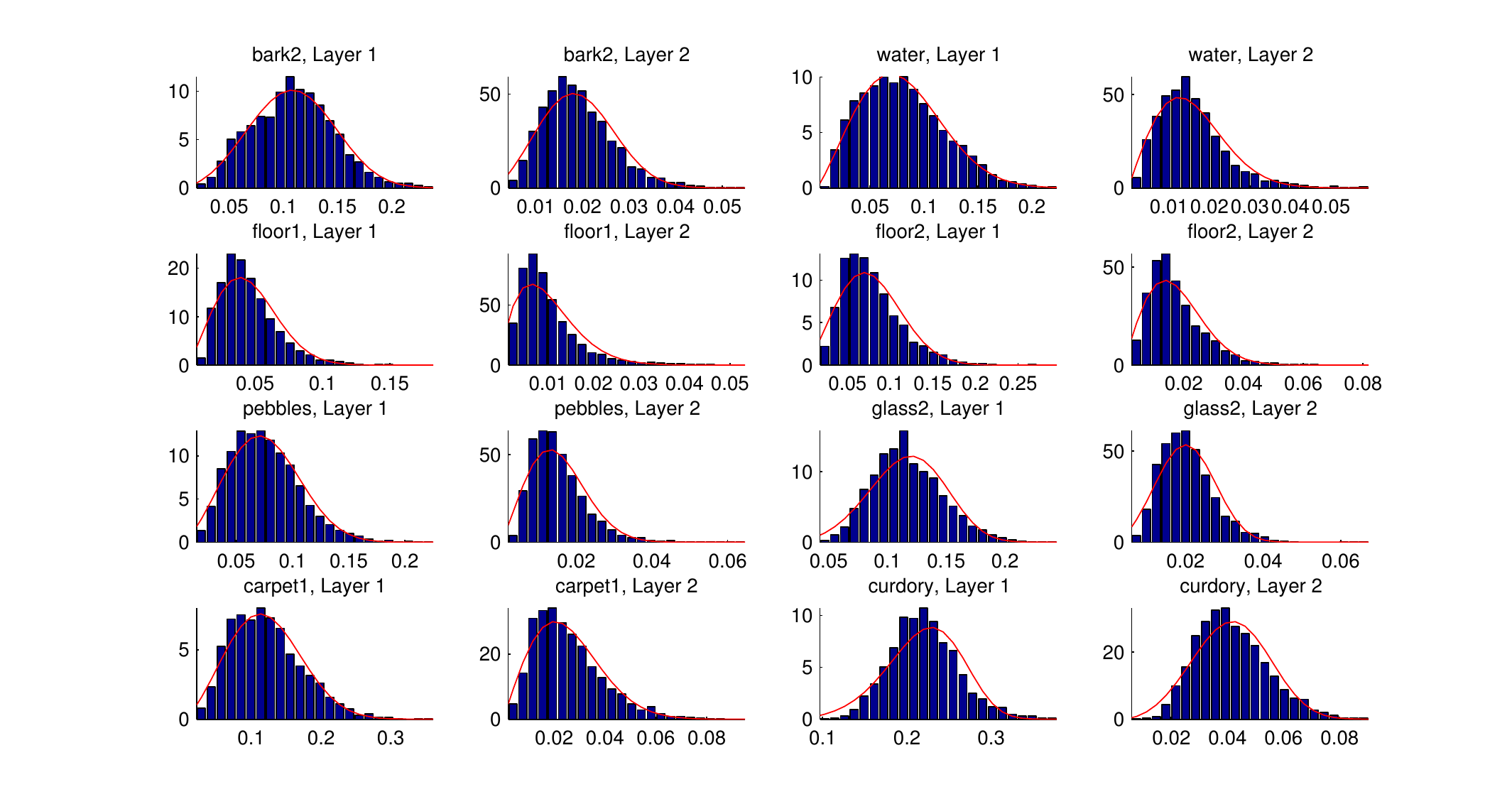}
   \caption{Histograms (blue) and respective ML Weibull fittings (red) of WST subbands from ayers $m=1$ and $m=2$ for different texture patches from the UIUC texture database}
  \label{fig:hists}
  \end{center}
\end{figure}

The above argument is purely empirical. The following discussion aims to justify further the choice of the WD as a model for WST subbands \eqref{eq:subband_m1} in the layers $m\geq1$: Like for the layer $m=0$, we neglect the impact of the low-pass filter $\phi_J$. The approximately analytical band-pass filters $\psi_{j,R}$ produce signals with similar zero-mean distributions in their real and imaginary parts. Moreover, it is known that band-pass components of a natural image can be well modeled by a GGD. This is also a reasonable assumption for band-pass components of Scattering subbands since they resemble natural images under dim lif´ghting conditions.  It is therefore natural to assume the GGD with the same parameters as a model for the distribution of the real and imaginary parts of the band-pass signals $\psi_{j_m,R_m}*|\cdots*|\psi_{j_1,R_1}*f|\cdots|$ involved in the WST. By neglecting statistical dependence between real and imaginary components, let us define a ``complex GGD'' with statistically independent real and imaginary parts as
\begin{equation}
p_\mathrm{CGGD}(z|\alpha,\beta)=p_\mathrm{GGD}(\Re(z)|\alpha,\beta)p_\mathrm{GGD}(\Im(z)|\alpha,\beta).
\end{equation}
In order to determine the PDF for the modulus $x=|z|\geq0$, we need to fix $|z|$ and integrate $p_\mathrm{CGGD}$ over the circles centered at 0 in the complex plane:
\begin{equation}
\begin{split}
&p_\mathrm{CGGD,abs}(x|\alpha,\beta)=\int_{|z|=x}p_\mathrm{CGGD}(z|\alpha,\beta)\mathrm{d}z\\
&=\left(\frac{\beta}{\alpha\Gamma(1/\beta)}\right)^2 x\int_{0}^{\frac{\pi}{2}}e^{-x^\beta((\cos\varphi)^\beta+(\sin\varphi)^\beta)/\alpha^\beta}\mathrm{d}\varphi
\end{split}
\label{eq:CGGD_abs}
\end{equation}
Evaluating the integral analytically is demanding, if not impossible. However, it is known that the modulus of a complex variable with statistically independent and Gaussian distributed real and imaginary parts abides the \textit{Rayleigh distribution}, which can be easily verified by substituting $\alpha=\sqrt{2}\sigma$ and $\beta=2$ into Eq.~\eqref{eq:CGGD_abs}, i.e.
\begin{equation}
p_\mathrm{CGGD,abs}(x|\sqrt{2}\sigma,2)=\frac{x}{\sigma^2}e^{x^2/(2\sigma^2)}.
\end{equation}
The usage of the GGD for modeling FWT subbands was motivated by the need to model histograms which are more or less peaked in shape than it is the case for the Gaussian distribution. For this purpose, the additional shape parameter $\beta$ was introduced. This corresponds to introducing an additional shape parameter to the Rayleigh distribution. Altering the peakedness of the real and imaginary part corresponds roughly to altering the skewness of the respective modulus distribution. From Eq. \eqref{eq:CGGD_abs}, we can observe that $p_\mathrm{CGGD,abs}(0|\alpha,\beta)=0$, so a good alternative to $p_\mathrm{CGGD,abs}$ for modeling the WST subbands would be a two-parameter generalization of the Rayleigh distribution which is controllable in its skewness without changing its value at 0. For $k>1$, these requirements are fulfilled by the WD. 

Even though it is sensible to model the WST with its Weibull coefficients, it is still questionable if this decision is justifiable for the NWST. It is certainly true for the first layer since it only involves an overall scaling. Unfortunately, this can not be assumed for the other descendant layers. Nevertheless, experiments show, that in practice this assumption still holds. However, for the most important ranges of $k$, the multiplicative inverses of WD distributed samples exhibit histograms which again can be well modeled by the WD as can be seen in Figure \ref{fig:hist_inv}. Thus, the normalized coefficients in \eqref{eq:nwst} for $m\geq 2$ are products of values whicht are close to be Weibull distributed and statistically independent. Thus the WD will again dominate the subband histograms.
\begin{figure}
\centering
\includegraphics[width=\columnwidth,clip,trim=15mm 10mm 10mm 10mm]{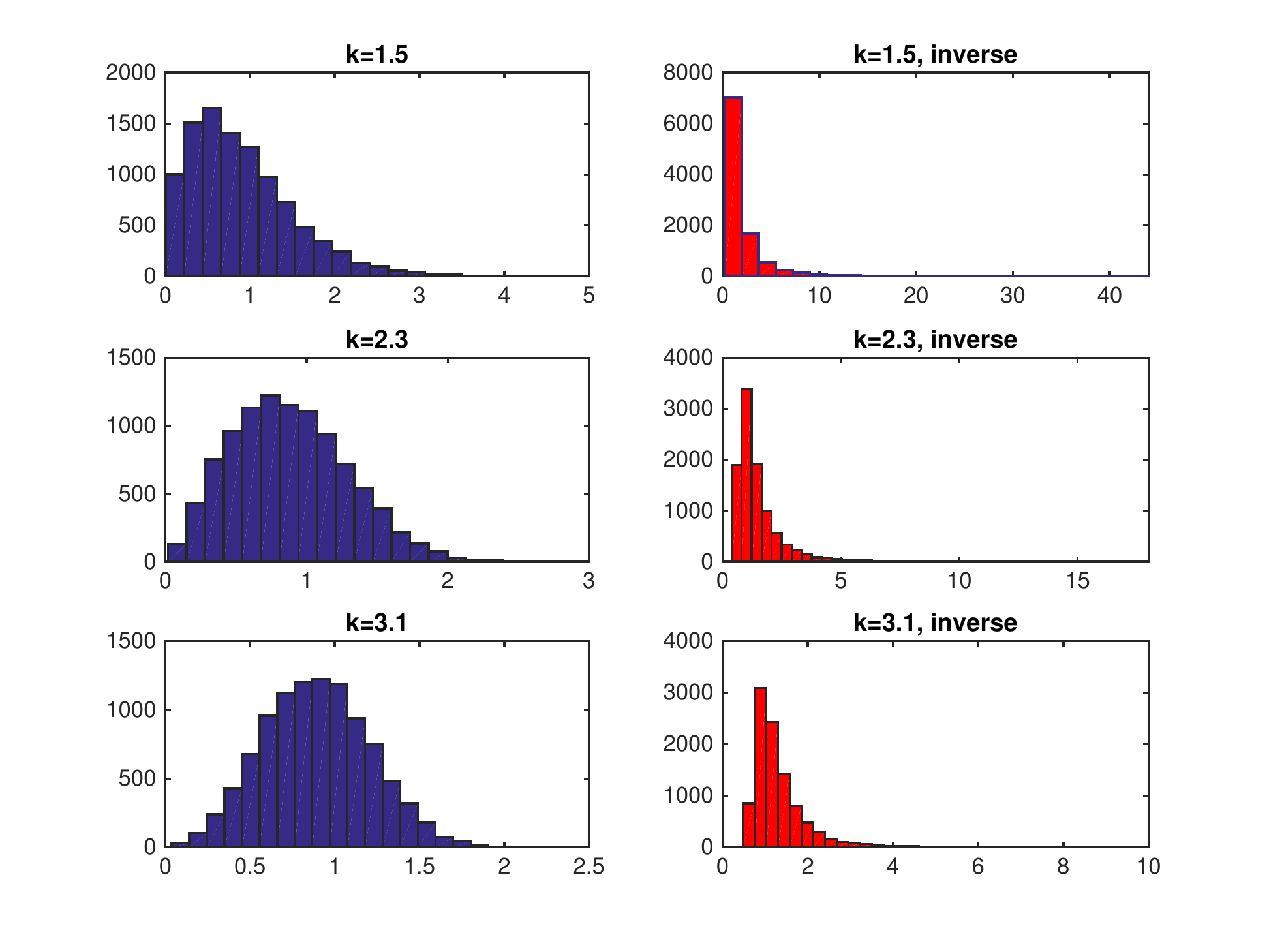}
\caption{Histogram of Weibull distributed samples and their inverses}
\label{fig:hist_inv}
\end{figure}

We employ a numerical approach to estimate the Weibull parameters via minimizing the corresponding ML function. The derivation of the algorithm can be found in Appendix~\ref{app:b}.


\subsection{Scattering Similarity}
\subsubsection{The Bhattacharyya Kernel}
Both the cross-entropy and the KLD can be derived for the WD in order to define an SM similar to the on discussed in Subsection  \ref{sub:SubHist}, cf. \cite{Bauckhage2013}. However, the derivation leads to very complex expressions. 

As an alternative, we propose an SM based on the \textit{Bhattacharyya} coefficient. For a pair of PDFs $p,q$, it is defined as
\begin{equation}
BC(p,q)=\int_{-\infty}^\infty \sqrt{p(\boldsymbol{x})q(\boldsymbol{x})}\mathrm{d}\boldsymbol{x}.
\label{eq:BC}.
\end{equation}
Expression \eqref{eq:BC} is a popular choice for Kernels in the machine learning community \cite{Jebara2004} and as such imposes a Hilbert space structure on probability based features.
\subsubsection{Weibull Similarity}
Our aim is to derive a simple approximation of \eqref{eq:BC} for a pair of Weibull distributions $p_\mathrm{WD}(x|\lambda_1,k_1),p_\mathrm{WD}(x|\lambda_2,k_2)$. To this end, let us assume that $k_1\approx k_2$. This is a justifiable assumption since the distributions are always compared for each subband individually. Let us further define
\begin{equation}
k=\frac{k_1+k_2}{2}\ \mathrm{and}\ \lambda=\sqrt[k]{\frac{\lambda_1^k+\lambda_2^k}{2}}.
\end{equation}
This ultimately leads to
\begin{equation}
\begin{split}
&BC(p_\mathrm{WD}(x|\lambda_1,k_1),p_\mathrm{WD}(x|\lambda_2,k_2))\\
=&\int_0^\infty \sqrt{p_{\mathrm{wbl}}(x|k_1,\lambda_1)*p_{\mathrm{wbl}}(x|k_2,\lambda_2)} dx \\
=&\sqrt{\lambda_1^{k_1}\lambda_2^{k_2}}\sqrt{k_1 k_2}\int_0^\infty x^{k-1}e^{-\frac{\lambda_1^{k_1} x^{k_1}+\lambda_2^{k_2}x^{k_2}}{2}} dx\\
\approx&\sqrt{\lambda_1^{k}\lambda_2^{k}}\sqrt{k_1 k_2}\int_0^\infty x^{k-1}e^{-(\lambda x)^k} dx\\
=&\sqrt{\lambda_1^k \lambda_2^k}\sqrt{k_1 k_2}\int_0^\infty \frac{\lambda^{-k}}{k}  p_\mathrm{wbl}(x|k,\lambda) dx\\
=&4\frac{\sqrt{\lambda_1^{k}\lambda_2^{k}}}{\lambda_1^{k}+\lambda_2^{k}}\frac{\sqrt{k_1 k_2}}{k_1+k_2}.
\end{split}
\label{eq:KerDrv}
\end{equation}
For the sake of convenience, let us write the arithmetical and geometrical mean of two values $y_1,y_2\in\mathbb{R}_+$ as
\begin{equation}
\mu_\mathrm{a}(y_1,y_2)=\frac{y_1+y_2}{2}\ \mathrm{and}\ \mu_\mathrm{g}(y_1,y_2)=\sqrt{y_1 y_2},
\end{equation}
respectively. From the last equation of \eqref{eq:KerDrv} we define our similarity measure for pairs of Weibull PDFs $p_\mathrm{WD}(x|\lambda_1,k_1),p_\mathrm{WD}(x|\lambda_2,k_2)$ as 
\begin{equation}
K(\lambda_1,k_1;\lambda_2,k_2)=\frac{\mu_g(\lambda_1^k,\lambda_2^k)}{\mu_a(\lambda_1^k,\lambda_2^k)}\cdot\frac{\mu_g(k_1,k_2)}{\mu_a(k_1,k_2)}.
\label{eq:WDKer}
\end{equation}
In the derivation of $K$, we replace $k_1$ and $k_2$ by $k$ each time it appears as the exponent of $\lambda_1$ and $\lambda_2$, respectively. Even though this step is not necessary to eveluate the integral in \eqref{eq:KerDrv}, it ensures a straight-forward understanding of \eqref{eq:WDKer}, with the first factor measuring the similarity of the scale and the second factor measuring the similarity of the shape. Moreover, just like \eqref{eq:BC}, the expression \eqref{eq:WDKer} defines a Kernel for Weibull PDFs, as it is shown in the following proposition.
\begin{prp}
Expression \eqref{eq:WDKer} is a Kernel for Weibull parameters.
\end{prp}
\begin{proof}
Note that for $k=k_1=k_2$ the approximation in \eqref{eq:KerDrv} becomes an equality. I.e.,
\begin{equation}
\begin{split}
BC(p_\mathrm{WD}(x|\lambda_1,k),p_\mathrm{WD}(x|\lambda_2,k))=&K(\lambda_1,k;\lambda_2,k)\\
=&\frac{\mu_g(\lambda_1^k,\lambda_2^k)}{\mu_a(\lambda_1^k,\lambda_2^k)}
\end{split}
\end{equation}
is a Kernel, since it is a Bhattacharyya coefficient of two PDFs. So is
\begin{equation}
BC(p_\mathrm{WD}(x|k_1,1),p_\mathrm{WD}(x|k_2,1))=\frac{\mu_g(k_1,k_2)}{\mu_a(k_1,k_2)},
\end{equation}
because both the shape and the scale parameter of the WD are constrained to be from the same domain $\mathbb{R^+}$. Therefore, Expression \eqref{eq:WDKer} is the product of two Kernels and as such a Kernel itself.
\end{proof}
The Kernel property of \eqref{eq:WDKer} provides us with some geometrical intuition, hence makes it a good choice as the fundamental building block of an SM. This includes the fact that it naturally induces a metric, which incorporates the notion of similarity as a geometric distance as is for instance the motivation behind multidimensional scaling \cite{Torgerson1958}.

It also enables the (N)WST to be conveniently subjected to many machine learning techniques such as Support Vector Machines.

\subsubsection{Similarity for Multiple Subbands}
Consider two PDFs $p,q$ of two statistically independent random variables $X,Y$,
\begin{equation}
p(x,y)=p_1(x)p_2(y),\;q(x,y)=q_1(x)q_2(y)
\end{equation}
It can be easily seen that
\begin{equation}
BC(p,q)=BC(p_1,q_1)BC(p_2,q_2).
\end{equation}
Carrying over this insight to the derivation \eqref{eq:KerDrv} yields, for a pair of sets of $N$ independent WDs with the parameter vectors $\boldsymbol{\lambda}^1,\boldsymbol{k}^1\boldsymbol{\lambda}^2,\boldsymbol{k}^2\in\mathbb{R}_+^N$,
\begin{equation}
K(\boldsymbol{\lambda}^1,\boldsymbol{k}^1;\boldsymbol{\lambda}^2,\boldsymbol{k}^2)=\prod_{i=1}^N K(\lambda_{1,i},k_{1,i};\lambda_{2,i},k_{2,i}).
\end{equation}
For a pair of WST transforms with $N$ subbands, it is therefore straightforward to derive an SM. Since by our definition a low SM indicates a high similarity we apply the logarithm and reverse the sign which finally leads us to
\begin{equation}
\begin{split}
&SM_\mathrm{Scat}(\boldsymbol{\lambda}^1,\boldsymbol{k}^1;\boldsymbol{\lambda}^2,\boldsymbol{k}^2)\\
=&-\sum_{i=1}^N\ln K(\lambda_{1,i},k_{1,i};\lambda_{2,i},k_{2,i}).
\end{split}
\label{eq:SM_scat}
\end{equation}

\subsection{Implementing the Framework}
The two previous sections provide the theoretical basis for building a complete CBIR based on WST Subband histograms.  In order to extract the features, the WST (Section \ref{sub:WST}) is performed on each image. Each subband of the layers $m\geq 1$ is then subjected to ML in order to extract the WD parameters. If the number of subbands is $N$, this amounts to feature vectors of size $2N$. In order to compare the query feature vector to the feature vectors in the database, the SM \eqref{eq:SM_scat} us employed.
As most of these steps can be performed independently of each
other, the whole system is well suited to be implemented in a parallel manner,
for example in applications for Big Data.

\section{Numerical Experiments}
\label{sec:exp}
\subsection{Experimental Settings}
\subsubsection{Implementation}
All experiments were performed in a Matlab 2014a environment. The code reproducing the key results is available online\footnote{https://www.ldv.ei.tum.de/uploads/media/texture\_retrieval\_scattering\_14.zip}. The WST was performed by means of the ScatNet\footnote{http://www.di.ens.fr/data/software/scatnet/} toolbox, version 0.2. The Weibull parameters were extracted via \texttt{wblfit} from the Statistical Toolbox.
Since it is the most standard subband histogram technique, we also implemented the FWT+GGD+KLD method according to \cite{Do2002} or comparison. Additionally, we implemented a method based the Dual-Tree Complex Wavelet Transform (DT-CWT), inspired by \cite{Kwitt2008}. The DT-CWT was performed by the DT-CWT Transform Pack\footnote{http://www-sigproc.eng.cam.ac.uk/Main/NGK}, version 4.3. Similarly the subband histograms were modeled by the WD, but the KLD was replaced by the proposed SM in \eqref{eq:SM_scat} for the sake of comparability.
\subsubsection{Choice of parameters}
Along with the regular WST, all experiments were also conducted with its normalized version as defined in \eqref{eq:nwst}. The directionality parameter, corresponding to the cardinality of the rotation group $\mathcal{R}$, is fixed to be $L=4$. All computations were performed with double precision and the subbands of the WST are critically sampled according to the Shannon-Nyquist sampling theorem. The bandwidth of the low-pass filter $\phi_J$ is chosen in a way such that it produces signals of size $16\times16=256$, in accordance with \cite{Do2002}. The WST is evaluated for the maximum path lengths $M=2$ and $M=3$. In theory, a linear increase in $M$ leads to an exponential increase in the number of subbands. However, since the modulus operation hardly traverses any energy toward higher frequencies, they become negligible with increasing path length. Apart from the mentioned, default settings of ScatNet were used. In particular, the Morlet Wavelet was used as the band-pass filter $\psi$ on each layer of the WST.

The FWT relies on the D2 4-tap Daubechies wavelet \cite{Mallat2008} as the underlying multiresolution representation. The DT-CWT was run with the options \texttt{antonini} and \texttt{qshift\_a}. In both cases, the size of the smallest subband was $16\times 16$, which corresponds to three levels of decomposition.
\subsubsection{Data}
\begin{figure*}
\begin{center}
\includegraphics[width=0.9\textwidth,clip,trim=0mm 80mm 0mm 80mm]{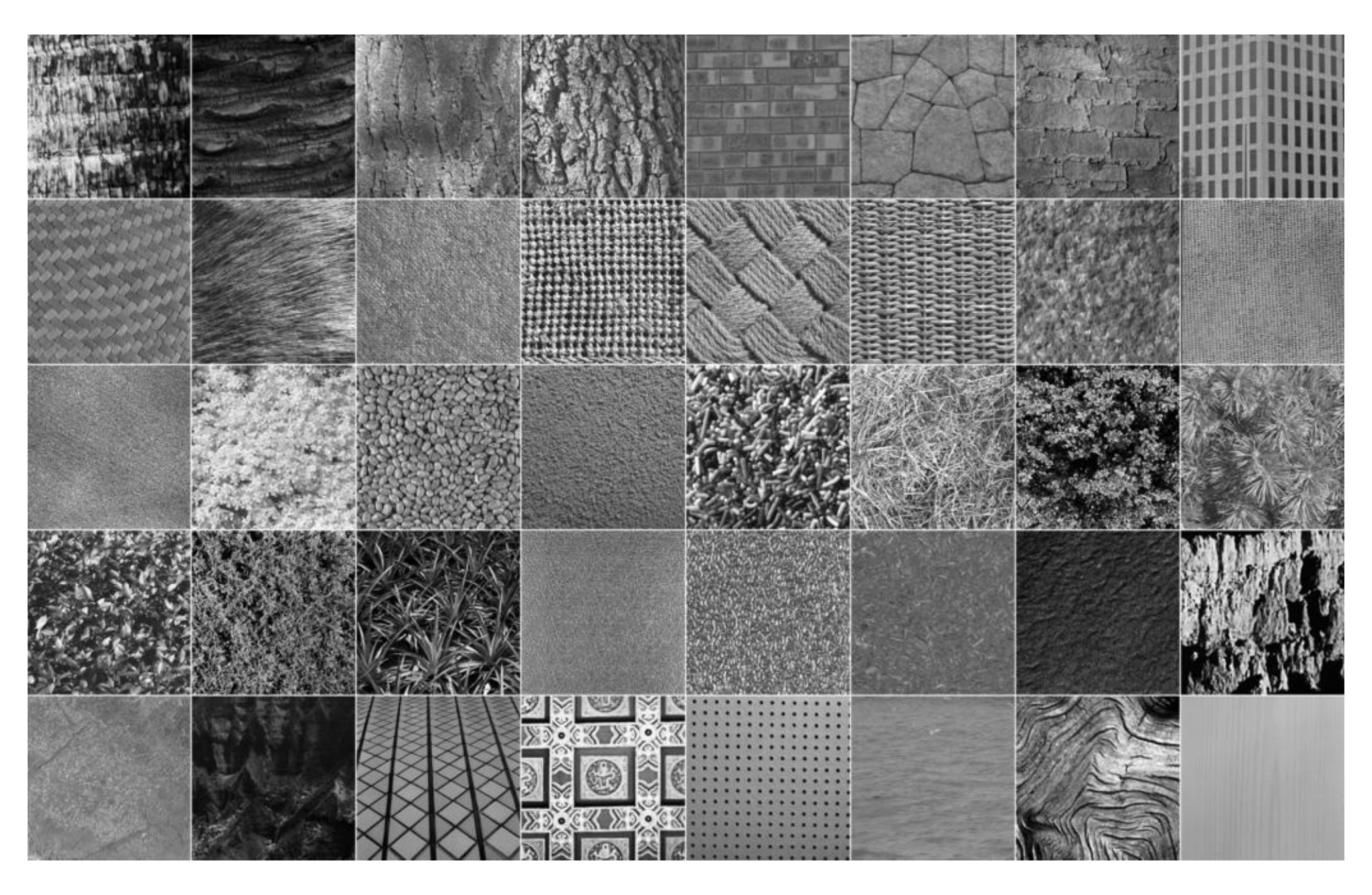}
\caption{Database D1 from left to right and top to bottom, with row and column in brackets: Bark0 (1,1), Bark6 (1,2), Bark8 (1,3), Bark9 (1,4), Brick1 (1,5), Brick4 (1,6), Brick5 (1,7), Buildings9 (1,8), Fabric0 (2,1), Fabric4 (2,2), Fabric7 (2,3), Fabric9 (2,4), Fabric11 (2,5), Fabric14 (2,6), Fabric15 (2,7), Fabric17 (2,8), Fabric18 (3,1), Flowers5 (3,2), Food0 (3,3), Food5 (3,4), Food 8 (3,5), Grass1 (3,6), Leaves8 (3,7), Leaves10 (3,8), Leaves11 (4,1), Leaves12 (4,2), Leaves16 (4,3), Metal0 (4,4), Metal2 (4,5), Misc2 (4,6), Sand0 (4,7), Stone1 (4,8), Stone4 (5,1), Terrain10 (5,2), Tile1 (5,3), Tile4 (5,4), Tile7 (5,5), Water5 (5,6), Wood1 (5,7), and Wood2 (5,8).}
\label{fig:vistex}
\end{center}
\end{figure*}
\begin{figure*}
\begin{center}
\includegraphics[width=0.9\textwidth,clip,trim=0mm 80mm 0mm 80mm]{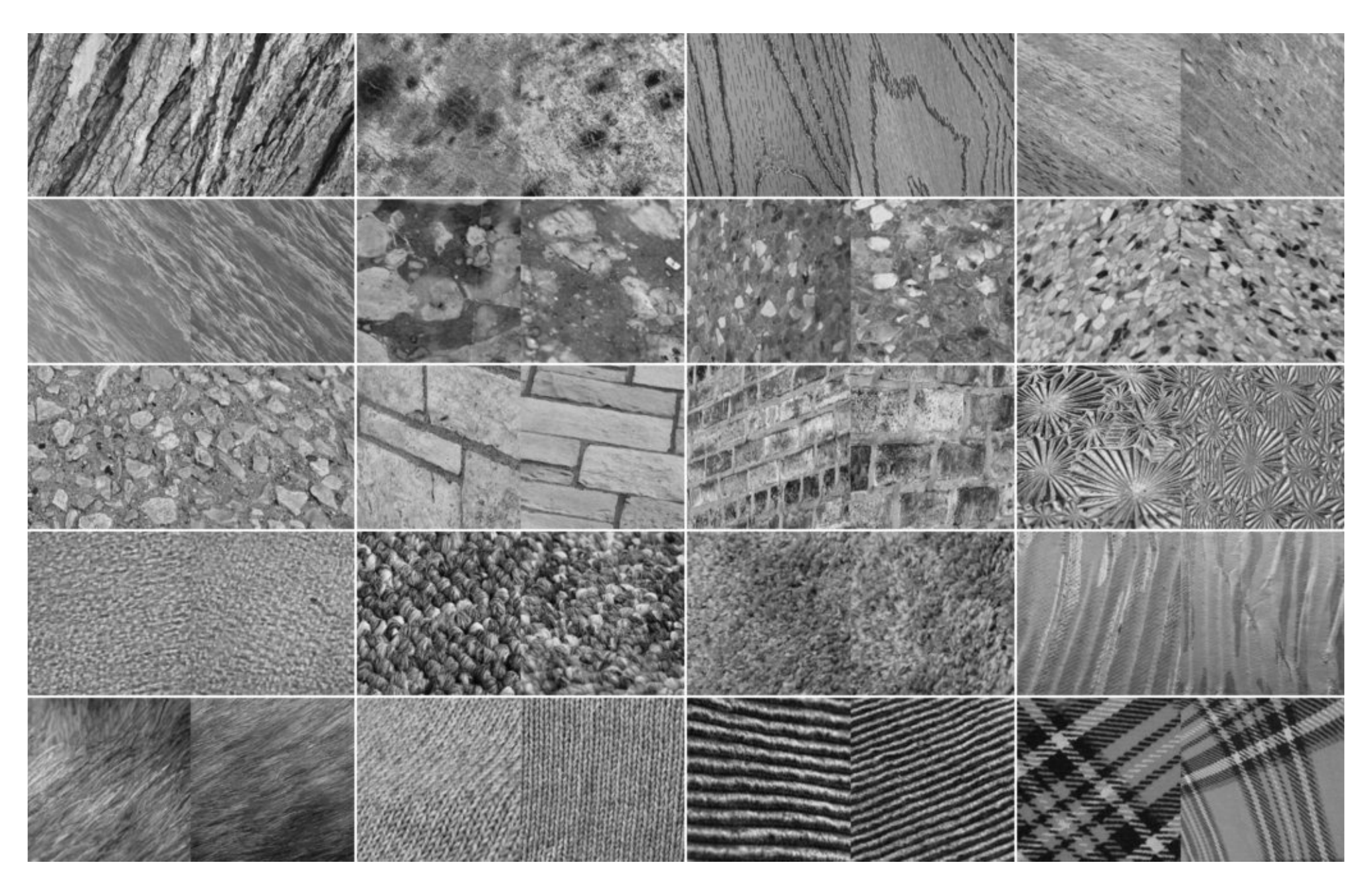}
\caption{Database D2 (patches) from left to right and top to bottom: Bark1, Bark2, Wood2, Wood3, Water, Marble, Floor1, Floor2, Pebbles, Wall, Brick1, Glass1, Glass2, Carpet1, Carpet2, Wallpaper, Fur, Knit, Curdoroy, Plaid.}
\label{fig:uiuc}
\end{center}
\end{figure*}

The database D1 is the same as used in \cite{Do2002} and is based on the VisTex database\footnote{http://vismod.media.mit.edu/vismod/imagery/VisionTexture/distribution.html}. It  consists of 40 texture images
of the size $512\times 512$, each divided into 16 non-overlapping $128\times 128$ patches. This amounts to 640 image samples of 40 different textures depicted in Fig. \ref{fig:vistex}. The database D2 was generated from images of the following subset of the UIUC texture database\footnote{http://www-cvr.ai.uiuc.edu/ponce\_grp/data/}: Bark1, Bark2, Wood2, Wood3, Water, Marble, Floor1, Floor2, Pebbles, Wall, Brick1, Glass1, Glass2, Carpet1, Carpet2, Wallpaper, Fur, Knit, Curdoroy, Plaid. Two $640\times 480$ images from each class are used to create five overlapping $256\times 256$ patches which are then scaled down to half the edge size. Hence, we get a database of 20 different texture classes each containing ten $128\times 128$ patches. Fig. \ref{fig:uiuc} depicts one patch from each utilized image. All image patches are normalized to zero mean and unit energy, in order to avoid any bias caused by the overall lighting condition of each original texture image. The set of all patches generated from the same texture is considered a class. Its cardinality will be denoted by $c$ in the following. Consequently, $c=16$ for D1 and $c=10$ for D2. For each image patch, the $c-1$ most similar patches were retrieved. The \textit{retrieval rate} for each patch is defined as the ratio of the number of retrieved patches from the same class to $c-1$. The overall retrieval rate is the average of the retrieval rates for all the images in the database.

\subsection{Results}
\subsubsection{Overview}
The retrieval rates are summarized in Table \ref{table:d1}.
\begin{table*}
\begin{center}
\begin{tabular}{| c | c | c | c | c | c | c|}
\hline
  & \specialcell{FWT+GGD\\+KLD} &  \specialcell{DT-CWT+WD\\+$SM_\mathrm{Scat}$} & \specialcell{WST+WD,$M$=2\\+$SM_\mathrm{Scat}$} & \specialcell{WST+WD,$M=3$\\+$SM_\mathrm{Scat}$} & \specialcell{NWST+WD,$M=2$\\+$SM_\mathrm{Scat}$} & \specialcell{NWST+WD,$M=3$\\+$SM_\mathrm{Scat}$}\\[0.2cm]
  \hline\hline
Database D1 & 75.50\% & 78.18\% & 78.93\%& 78.14\% & 84.90\% & \textbf{85.30\%}\\
\hline
Database D2 & 52.39\% & 59.61\% & 66.50\% & 63.78\% & \textbf{66.94}\% & 65.17\% \\
\hline
\end{tabular}
\end{center}
\caption{Performance of WST + WD and NWST + WD in comparison with FWT + GGD and DT-CWT + GGD on databases 1 and 2}
\label{table:d1}
\end{table*}
\begin{table}
\begin{center}
\begin{tabular}{|c | c | c | c|}
\hline
   \specialcell{FWT+GGD\\+KLD} &  \specialcell{FWT+G$\Gamma$D\\+KLD} & WD-HMM & Rotated Wavelets\\
  \hline\hline
  76.93\% & 78.40\%  & 80.05\% & 82.81\% \\
  \hline
\end{tabular}
\end{center}
\caption{Performance of state of the art methods on Database D1}
\label{table:d1_ref}
\end{table}
While WST+WD+$SM_\mathrm{Scat}$ produces a similar result as DT-CWT+WD+$SM_\mathrm{Scat}$ on Database D1, NWST+WD+$SM_\mathrm{Scat}$ is able to outperform all of the competing frameworks by $4.72\%$ for $M=2$ and $5.12\%$ for $M=3$. Database D1 is widely used as a benchmark for CBIR retrieval. In order to provide a sense for the state of the art, Table \ref{table:d1_ref} summarizes the results from recent publications on comparable approaches: DWT + Generalized Gamma Distribution (G$\Gamma$D) \cite{Choy2010}, Wavelet Domain Hidden Markov Models (WD-HMM) \cite{Do2002a} and Rotated Complex Wavelets \cite{Kokare2005}. Also, the original result for FWT+GGD+KLD from \cite{Do2002} was included, since we were not able to reproduce it. To the authors' best knowledge, the best published result for this experiment so far was produced by Rotated Complex Wavelets with a retrieval rate of $82.81\%$ which is still outperformed by 2.09\% by NWST+WD+$SM_\mathrm{Scat}$ with $M=2$ and 2.49\% with $M=3$. 

Database D2 is considerably smaller than D1, but involves more variation within the classes, for instance, in terms of camera angle and deformation. Again, the WST greatly improves the retrieval performance. However, this time the regular WST does not fall behind the NWST. Also, increasing the maximum path length up to $M=3$ harms the performance. However, we can conclude that NWST+WD+$SM_\mathrm{Scat}$ performs comparably well and produces the best results in both our test settings.

\subsubsection{Impact of Directionality}
In general, natural textures contain strongly directional features. Hence, it is often beneficial to incorporate a frame of higher directionality as it is done for the Rotated Wavelets approach, for example. The DT-CWT+WD+$SM_\mathrm{Scat}$ method which mostly draws its benefits from increased directional selectivity in comparison with standard FWT based techniques is consistently and significantly outperformed by NWST+WD+$SM_\mathrm{Scat}$, despite the same directionality properties. That is to say, there are reasons to believe that the success of the WST based methods is not exclusively due to increased directionality.

\subsubsection{Impact of Textural Structure}
\begin{figure*}
\begin{center}
\includegraphics[width=\textwidth,clip,trim=70mm 0mm 20mm 15mm]{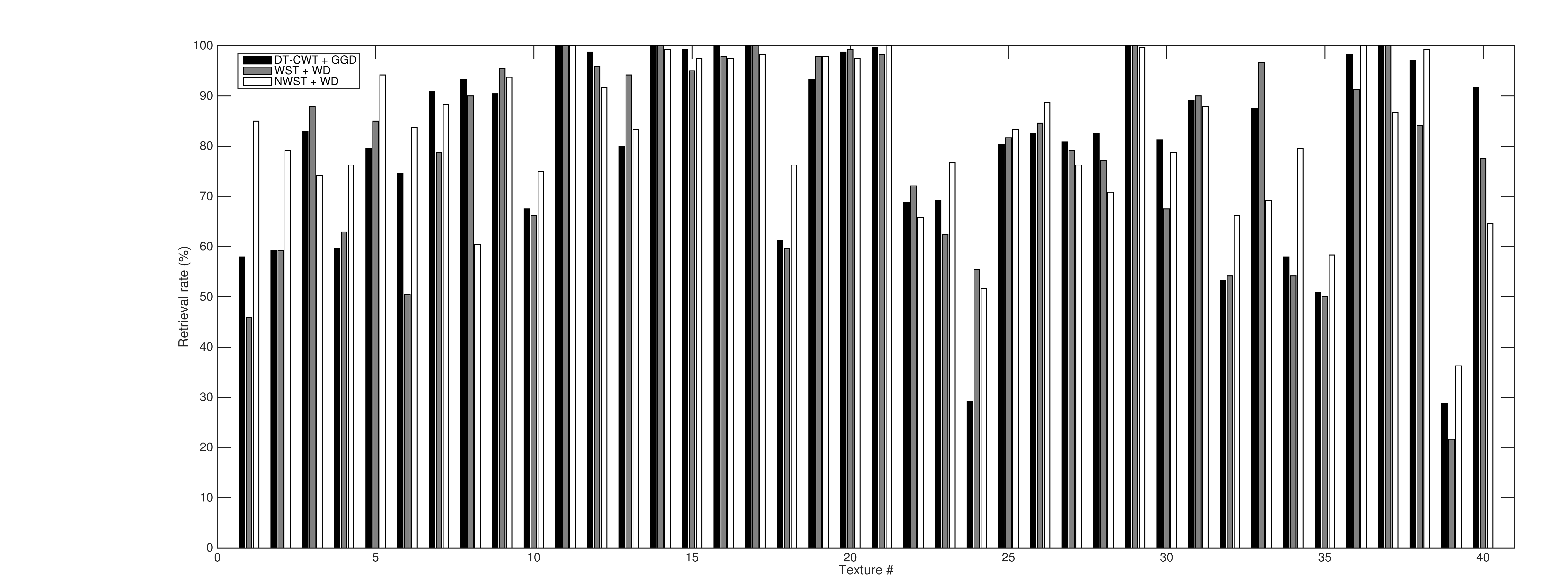}
\caption{Retrieval rates for each texture class of D1 individually. From left to right, with index in brackets: Bark0 (1), Bark6 (2), Bark8 (3), Bark9 (4), Brick1 (5), Brick4 (6), Brick5 (7), Buildings9 (8), Fabric0 (9), Fabric4 (10), Fabric7 (11), Fabric9 (12), Fabric11 (13), Fabric14 (14), Fabric15 (15), Fabric17 (16), Fabric18 (17), Flowers5 (18), Food0 (19), Food5 (20), Food8 (21), Grass1 (22), Leaves8 (23), Leaves10 (24), Leaves11 (25), Leaves12 (26), Leaves16 (27), Metal0 (28), Metal2 (29), Misc2 (30), Sand0 (31), Stone1 (32), Stone4 (33), Terrain10 (34), Tile1 (35), Tile4 (36), Tile7 (37), Water5 (38), Wood1 (39), and Wood2 (40). }
\label{fig:bars}
\end{center}
\end{figure*}
\begin{figure}
\begin{center}
\includegraphics[width=\columnwidth,trim=25mm 10mm 25mm 10mm, clip]{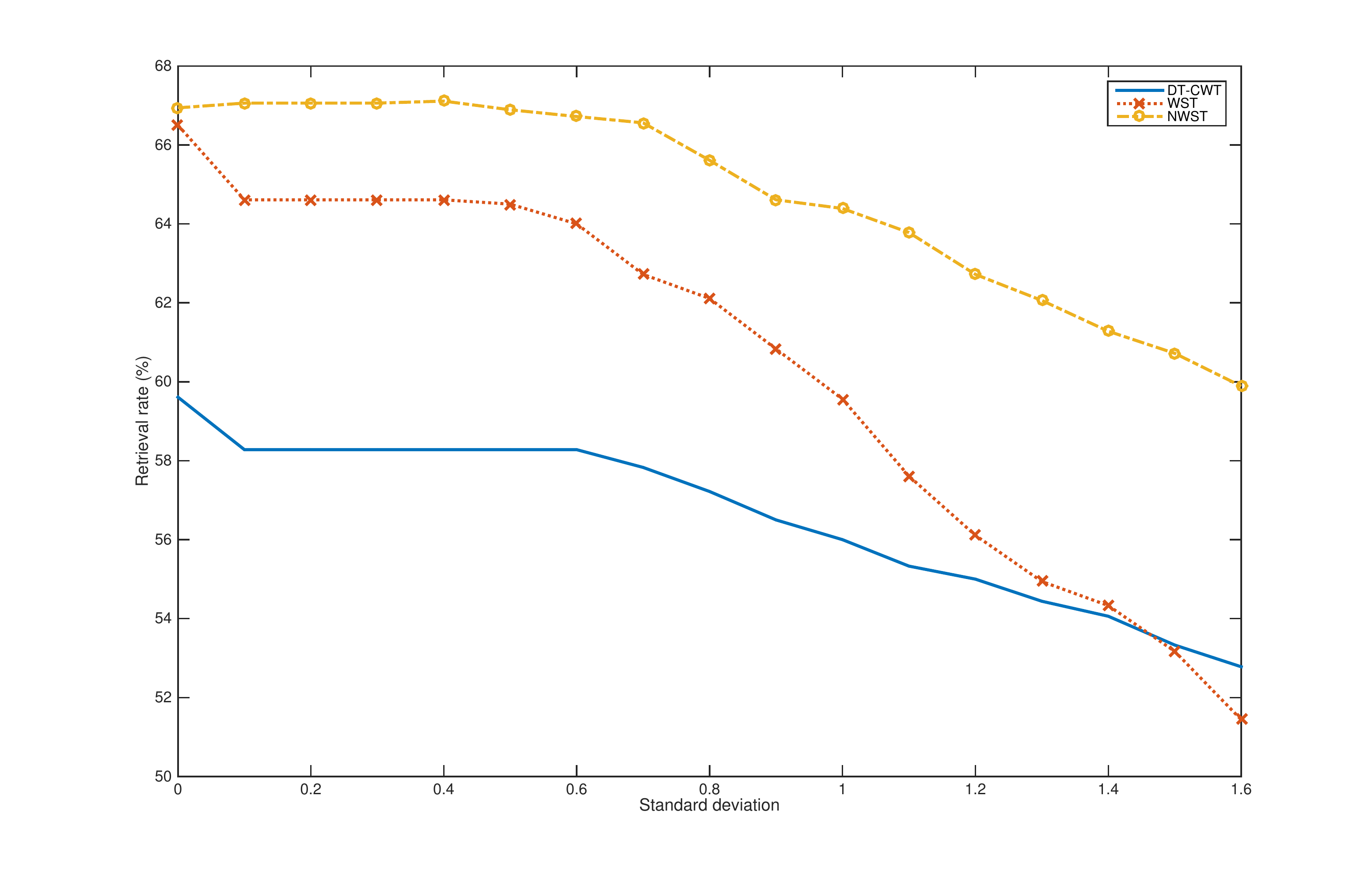} \vspace*{-6mm}
\caption{Retrieval Rate}
\label{fig:blur}
\end{center}
\end{figure}
To further analyze the capabilities of the involved methods, consider Fig. \ref{fig:bars}. It shows the retrieval rates for each class in Database D1 produced by DT-CWT as well as by WST and NWST with $M=2$. DT-CWT + GGD performed well in comparison to the WST based techniques for the classes Brick5, Buildings9, Metal0 and Wood2 and badly for Bark9, Fabric4, Leaves10 and Stone1. In general, it appears that DT-CWT is most suitable for flat and (piecewise) smooth surfaces, while the two methods involving Scattering seem to be preferable for surfaces that are rather uneven, be it because they are rough or because they are bent or dented. This correlates with the premise that Scattering coefficients are less sensitive to spatial deformations, since on the one hand, images of rough surfaces carry significant fractions of high-frequency components and on the other hand spatial deformation in the image is often caused by locally bending or denting the depicted material. More specifically, NWST performed comparably well for the classes Bark0, Bark6, Flowers5, Terrain10 and regular WST for the classes Bark8, Fabric11, Grass1, Stone4. This suggests that WST works better on fine structures while NWST yields better results for coarse structures.

This assumption is corroborated by the following experiment. Fig. \ref{fig:blur} depicts the retrieval rates for the Database D2 after blurring it with a Gaussian filter of different bandwidths. Unsurprisingly, the retrieval performance declines with stronger blurring for all methods. However, it is evident that the retrieval rate for WST declines significantly faster than the other two. This can be interpreted as a consequence of the distortion invariance of normalized Scattering coefficients. More importantly, this confirms that the regular WST method relies on fine features since they are most heavily affected by blurring. On the other hand, images of coarse patterns and flat and piecewise smooth surfaces are less affected by blurring such that the other two methods appear to be relatively robust.

\section{Conclusion and Outlook}
\label{sec:conclusion}
In this work, we propose a subband histogram FE method based on the statistics of the WST of a texture image. Our derivation and analysis demonstrate how to extract statistical features from a WST representation of a texture image by means of matching a Weibull model via ML and how to 
measure 
the similarity of the respective feature vectors via the KLD. The proposed techniques come in handy when it comes to reducing the enormous degree of redundancy introduced by the WST since they represent each subband by as much as two real numbered parameters. The experiments show that the proposed method can outperform comparable algorithms based on filter bank transforms in terms of retrieval accuracy when designed as a CBIR system. Regular WST coefficients seem to work better on fine structures and Normalized WST coefficients on coarse structures.

Since approaches for rotation invariant Scattering representations are already available \cite{Sifre2012, Sifre2013}, it appears feasible to extend the presented ideas towards a rotation invariant CBIR framework which would allow for more general problem settings. 

The summation in \eqref{eq:SM_scat} was partly motivated by assuming statistical independence of the subbands. However, neither does statistical independence hold for the WST subbands, nor for the subbands of other filterbank transforms, in practice. Accounting for statistical dependence in the model could lead to significant improvements in efficiency \cite{Tzagkarakis2006,Do2002a}.

\appendix

\subsection{Estimation of Weibull parameters}
\label{app:b}
In this section, we employ a numerical optimization approach to
estimate the Weibull parameters. A more detailed derivation can be found in \cite{Sornette2006}.
The log-likelihood amounts to
\begin{equation}
\begin{split}
 f_{WD}(k,\lambda) = &\ \ln L_{WD}(x_1,\dots,x_N|k,\lambda) \\
= &~ Nk\ln\lambda+N\ln k\\
&+(k-1)\sum_{i=1}^{N}\ln x_i -\sum_{i=1}^{N}(\lambda x_i)^k.
\end{split}
\label{eq:like_wibl}
\end{equation}
The optimal parameters $(k^{*},\lambda^{*})$ are identified by solving
\begin{eqnarray}
	(k^{*},\lambda^{*}) = \operatorname*{argmax}_{(k,\lambda)} f_{WD}(k,\lambda).
\end{eqnarray}
The first derivatives of $f_{WD}$ with respect to $k$ and $\lambda$ are computed
as
\begin{equation}
\begin{split}
\frac{\partial f_{WD}}{\partial k}=&\ \frac{N}{k}+\sum_{i=1}^{N}\ln x_i+\left(N-\lambda^k\sum_{i=1}^{N}x_i^k\right)\ln\lambda 
\\&-\lambda^k\sum_{i=1}^{N}(\ln x_i)x_i^k
\end{split}
\label{eq:dL_dk}
\end{equation}
and
\begin{equation}
\frac{\partial f_{WD}}{\partial \lambda}=k\left(\frac{N}{\lambda}-\lambda^{k-1}\sum_{i=1}^{N}x_i^k\right).
\label{eq:dL_dl}
\end{equation}
Setting Eq.~\eqref{eq:dL_dl} to be equal to $0$ uniquely determines $\lambda^*$ for $k>0$. Solving for $\lambda$ yields
\begin{equation}
{\lambda^*} = \left(\frac{1}{N} \sum_{i=1}^N x_i^k\right)^{-\frac{1}{k}}.
\label{eq:lambda}
\end{equation}
Substituting \eqref{eq:lambda} in Eq.~\eqref{eq:dL_dk} and equating it with 
$0$ characterize the critical points of $f_{WD}$ in terms of $k$ for $\lambda=\lambda^*$, i.e.
\begin{equation}
\frac{N}{k}+\sum_{i=1}^{N}\ln x_i-\frac{N\sum_{i=1}^{N}(\ln x_i)x_i^k}{\sum_{i=1}^{N}x_i^k}=0.
\label{eq:k_solve}
\end{equation}
Since the function $f_{WD}$ as defined in Eq.~\eqref{eq:like_wibl} is concave with respect to
$k>0$, a numerical algorithm such as the Newton-Raphson method can be used to obtain $k^{*}$.
Substituting this solution into Eq.~\eqref{eq:lambda} yields $\lambda^{*}$.

\ifCLASSOPTIONcaptionsoff
  \newpage
\fi



%

\bibliographystyle{IEEEtran}
\bibliography{IEEEtran}




\end{document}